\documentclass[runningheads,envcountsame]{llncs}
\usepackage[T1]{fontenc}

\usepackage{amsfonts,amssymb,amsmath}
\usepackage{mathtools}
\usepackage{mathrsfs}
\usepackage{tikz}
\usepackage{tikz-cd}
\usepackage{enumitem}

\usepackage{graphicx}
\usepackage{color}
\usepackage{hyperref}
\usepackage{stackengine}

\usetikzlibrary{arrows, trees}

\makeatletter
\renewcommand{\tableofcontents}{%
  \section*{\contentsname}
  \@starttoc{toc}
}
\makeatother

\newenvironment{prf}{\begin{proof}}{\qed\end{proof}}

\spnewtheorem{dfn}[theorem]{Definition}{\bfseries}{}

\newcommand{\zero}{\mathbf{0}} 
\newcommand{\colim}{\operatornamewithlimits{colim}}
\newcommand{\id}{\mathrm{id}} 
\newcommand\EM{\mathsf{EM}}
\newcommand\Id{\mathrm{Id}}
\newcommand\Pow{\mathcal{P}}

\newcommand{\down}{{\downarrow}} 

\newcommand\nattrans{\Rightarrow}
\tikzset{
    nat/.style  = {-implies,double,double equal sign distance}
}
\newenvironment{tikzcdnat}[1][]{\begin{tikzcd}[arrows=nat,#1]}{\end{tikzcd}}

\newcommand{\Q}{\mathscr{Q}} 
\newcommand{\M}{\mathscr{M}} 

\newcommand{\emb}{\rightarrowtail} 
\newcommand{\epi}{\twoheadrightarrow} 
\newcommand{\xepi}[2][]{%
  \xrightarrow[#1]{#2}\mathrel{\mkern-14mu}\rightarrow
}
\newcommand{\xemb}[2][]{\ensurestackMath{\mathrel{%
  \stackengine{1pt}{%
    \stackengine{0pt}{\rightarrowtail}{\scriptstyle#2}{O}{c}{F}{F}{S}%
  }{\scriptstyle#1}{U}{c}{F}{F}{S}%
}}}

\DeclareMathOperator{\Set}{\mathsf{Set}}
\DeclareMathOperator{\C}{\mathscr{C}} 
\DeclareMathOperator{\A}{\mathscr{A}} 
\DeclareMathOperator{\T}{\mathscr{T}} 

\DeclareMathOperator{\Emb}{\mathbb{S}} 
\DeclareMathOperator{\Path}{\mathbb{P}} 
\newcommand{\U}{\mathcal{U}} 
\DeclareMathOperator{\htf}{\mathrm{ht}} 

\newcommand{\Ek}{\mathbb{E}_k} 
\newcommand{\Mk}{\mathbb{M}_k} 
\newcommand{\Pk}{\mathbb{P}_k} 
\newcommand\CC{\mathbb C}

\newcommand{\point}{\bullet}
\newcommand{\sg}{\sigma}
\newcommand{\R}{\mathsf{Str}}
\newcommand{\RS}{\R(\sg)}
\newcommand{\pRS}{\R_{\point}(\sg)}

\newcommand{\pp}{\mathbin{+\mkern-10mu+}}

\newcommand{\ML}{\mathrm{ML}} 

\newcommand{\Fraisse}{Fra\"{i}ss\'{e}}
\newcommand{\EF}{Ehren\-feucht--\Fraisse}
\renewcommand{\epsilon}{\varepsilon}
\renewcommand{\theta}{\vartheta}
\renewcommand{\phi}{\varphi}

\newcommand\ete[1]{\enspace\text{#1}\enspace}
\newcommand\qtq[1]{\quad\text{#1}\quad}
\newcommand\qq[1]{\quad{#1}\quad}

\newcommand\sue{\subseteq}
\newcommand\ar{\mathrm{ar}} 

\newcommand{\w}[1]{\widehat{#1}}
\renewcommand{\o}[1]{\overline{#1}}

\newcommand\sqleq{\sqsubseteq}

\tikzcdset{
    relay arrow/.default = 10pt,
    relay arrow/.style = {
        rounded corners,
        to path = {
            \pgfextra{
                \def\sourcecoordinate{\pgfpointanchor{\tikztostart}{center}}
                \def\targetcoordinate{\pgfpointanchor{\tikztotarget}{center}}
                \pgfmathanglebetweenpoints{\sourcecoordinate}{\targetcoordinate}
                \edef\tempangle{\pgfmathresult}
                \pgftransformrotate{\tempangle}
                \pgfmathifthenelse{#1>0}{\tempangle+90}{\tempangle-90}
                \pgfcoordinate{tempcoord}{\pgfpointanchor{\tikztostart}{\pgfmathresult}}
            }
            (tempcoord)
            -- ([yshift=#1]tempcoord)
            -- ([yshift=#1]tempcoord-|\tikztotarget.center)\tikztonodes
            --(\tikztotarget)
        }
    }
}

\begin{document}
\title{On the Axioms of Arboreal Categories\thanks{This is an extended version of a paper due to appear in the proceedings of the 18th International Workshop on Coalgebraic Methods in Computer Science (CMCS~2026).}}
\author{
Tom\'a\v s Jakl\inst{1}
\and
Luca Reggio\inst{2}}
\authorrunning{T. Jakl, L. Reggio}
\institute{Czech Technical University, Prague, Czech Republic
\\
\email{tomas.jakl@cvut.cz}
\quad
\url{https://tomas.jakl.one}
\\[1em]
\and
Universit\`a degli Studi di Milano, Milan, Italy
\\
\email{luca.reggio@unimi.it}
\quad
\url{https://lucareggio.github.io/}
}
%

\let\oldaddcontentsline\addcontentsline
\renewcommand{\addcontentsline}[3]{}
\let\oldaddtocontents\addtocontents
\renewcommand{\addtocontents}[2]{}

\maketitle

\let\addcontentsline\oldaddcontentsline
\let\addtocontents\oldaddtocontents

%
\begin{abstract}
Arboreal categories were introduced as an axiomatic framework for game comonads, which provide a comonadic view on many model-comparison games in logic. We demonstrate the inadequacy of the axiom stating that paths are connected. We then propose the notion of ``tree-connectedness'' to address this deficiency, and show that all the essential properties of arboreal categories that we are aware of remain valid under this new definition. Furthermore, we show that the path functor is a Street fibration.
\end{abstract}

\setcounter{tocdepth}{1}
\tableofcontents

\section{Introduction}
\label{s:intro}

Game comonads were introduced in \cite{abramsky2017pebbling,AbramskyS18} to give a categorical and comonadic account of several concepts central to finite model theory.
The key insight is that model-comparison games, a basic tool of (finite) model theory,
can be naturally organised into endofunctors on categories of relational structures that carry a comonad structure.
In fact, these endofunctors encode the possible plays or positions in a given relational structure according to the rules of the game, and the relational structure on the set of plays is derived from the winning conditions.
The principal examples are:
\begin{enumerate}[label=(\roman*)]
\item\label{i:pebble} Pebble games, which capture equivalence in the finite-variable fragments of infinitary first-order logic $\mathcal{L}_{\infty,\omega}$; these correspond to \emph{pebbling comonads}~\cite{abramsky2017pebbling}. 
\item\label{i:EF} Ehrenfeucht--Fra\"iss\'e games, which capture equivalence in the fragments of ordinary first-order logic $\mathcal{L}_{\omega,\omega}$ with bounded quantifier rank; these correspond to the \emph{Ehrenfeucht--Fra\"iss\'e comonads} introduced in~\cite{AbramskyS18,AS2021}.
\item\label{i:modal} Bisimulation games, which capture equivalence in (the standard translation of) modal logic with bounded modal depth; these correspond to the \emph{modal comonads} introduced in \emph{op.~cit}.
\end{enumerate}

Over the years, the term \emph{game comonads} has become established for comonads that describe model-comparison games. Further examples include hybrid logic~\cite{AM2022}, restricted conjunction logic~\cite{MS2022,schindling2025requantification}, description logic~\cite{BednarczykUrbanczyk2022comonadicdescriptionlogics}, logics with generalised quantifiers~\cite{conghaile2021game} or guarded quantifiers~\cite{Guarded2021}, and path predicate logic~\cite{FigueiraGoren2025}.

The category of Eilenberg--Moore coalgebras is central to the theory.
First, its structure is used to link game comonads and the corresponding model-comparison games.
Second, coalgebras encode decompositions of relational structures, yielding a correspondence with important combinatorial parameters such as tree depth and tree width, which has been exploited e.g.\ in~\cite{DJR2021,Reggio2021polyadic,abramskyjaklpaine2022}.

To capture the common features of the categories of Eilenberg--Moore coalgebras of game comonads, the axiomatic framework of arboreal categories was proposed in \cite{AR2021icalp,AR2022}.
The axiomatisation builds on the essential insight that cofree coalgebras encode the plays according to the rules of the game and that these plays can be identified externally as embeddings of finite ``linearly ordered'' coalgebras into cofree coalgebras.

Finite linearly ordered coalgebras are abstractly defined as \emph{paths}, i.e.\ objects with a finite chain of subobjects, according to a chosen factorisation system.
The axioms of arboreal categories ensure that paths are well behaved and that a general notion of bisimilarity (due to Joyal, Nielsen and Winskel~\cite{JNW1993}) is equivalent to an abstract notion of back-and-forth game played between \emph{path embeddings}.

In Section~\ref{s:arboreal-old} we recall the definition of arboreal category and show that, while its axioms are sound in the case of items~\ref{i:pebble}-\ref{i:EF} above, one of the axioms is not satisfied by (the coalgebras for) the modal comonads in item~\ref{i:modal}. In fact, the axiom fails for comonads over relational structures with constants, such as pointed Kripke frames.

The problematic axiom states that paths are connected in the sense of Definition~\ref{def:arboreal-cat} below.
This axiom is satisfied in the case of the pebbling and Ehrenfeucht--Fra\"iss\'e comonads, where coalgebras are \emph{forest-ordered} structures that satisfy appropriate properties. The notion of connected object involves considering coproducts, which, at the level of the underlying forest orders, are obtained by forming a disjoint union. In contrast, the coalgebras for the modal comonads are \emph{tree-ordered}; that is, the underlying forest orders have a bottom element. In particular, when forming a coproduct, the bottom elements must be identified. This is essentially what causes the connectedness axiom to fail, as seen in Example~\ref{ex:failure-connectedness}.

In Section~\ref{s:arboreal-new} we amend the definition of arboreal category by replacing the problematic notion of connectedness with a newly introduced notion of \emph{tree-connectedness}.
The main properties of arboreal categories established in \cite{AR2021icalp,AR2022} using the connectedness axiom remain valid if we substitute connectedness with tree-connectedness.
For further details, see Section~\ref{s:posets-subobjects}. We believe that the newly proposed notion of arboreal category (see Definition~\ref{def:arboreal-cat-new}) not only solves an existing deficiency but also leads to a robust structure theory:
\begin{itemize}
\item All examples (known to the authors) of categories of coalgebras for game comonads that are arboreal according to the definition in~\cite{AR2021icalp} are also arboreal according to the new definition.
\item Coalgebras for the modal comonads form arboreal categories according to the new definition (but not according to the one in~\cite{AR2021icalp}).
\item The \emph{path functor} from an arboreal category\footnote{In this and the next bullet point, ``arboreal category'' refers to the new definition.} to the category of trees is a fibration, as we show in Section~\ref{s:path-func}.
\item A broad class of arboreal categories satisfies a representation result. For more information on this and other aspects of the structure theory of arboreal categories, please refer to the ``Future work'' section at the end of the paper.
\end{itemize}

The purpose of this paper is two-fold. First, to introduce the new connectedness axiom and explain where the old axioms are lacking. Second, to introduce the theory of game comonads and arboreal categories to the universal coalgebra community; to this end, we present our main examples from this perspective.

\vspace{1em}
\noindent\emph{Relation to the existing literature on arboreal categories.} Since the introduction of arboreal categories in~\cite{AR2021icalp}, a number of articles on or related to them have appeared. These include~\cite{AR2022,AR2024,reggio2023finitely,ALR2025}. All facts concerning arboreal categories established in these papers remain valid under the newly proposed notion of arboreal category. Section~\ref{s:posets-subobjects} provides more details on how the relevant proofs need to be adapted. The paper \cite{AMS2024} on linear arboreal categories requires further inspection.

\section{Preliminaries}
\label{s:prelim}

We assume the reader is familiar with the basic notions of category theory, such as functors, natural transformations, adjunctions, limits and colimits (see e.g.~\cite{abramskytzevelekos2010introduction} or \cite{awodey2010category}).
In this section, we recall some basic facts about comonads, coalgebras, and factorisation systems.
While the technical focus of the paper is not on comonads, an understanding of this topic is necessary to appreciate the main examples.

\subsection{Comonads and coalgebras}
A \emph{comonad} on a category $\C$ is a tuple $(\CC,\epsilon,\delta)$ where $\CC\colon \C \to \C$ is an endofunctor, and the \emph{counit} $\epsilon\colon \CC \nattrans \Id$ and \emph{comultiplication} $\delta\colon \CC \nattrans \CC^2$ are natural transformations such that the following diagrams commute.
\[
    \begin{tikzcdnat}
        \CC \rar{\delta}\dar[swap]{\delta} & \CC^2 \dar{\delta \CC}  \\
        \CC^2 \rar[swap]{\CC\delta} & \CC^3
    \end{tikzcdnat}
    \qquad
    \begin{tikzcdnat}
        \CC \rar{\delta} \ar{dr}[swap]{\id} & \CC^2 \dar{\epsilon \CC} \\
        & \CC
    \end{tikzcdnat}
    \qquad
    \begin{tikzcdnat}
        \CC \rar{\delta} \ar{dr}[swap]{\id} & \CC^2 \dar{\CC\epsilon} \\
        & \CC
    \end{tikzcdnat}
\]

Given a functor \(F\colon \C \to \C\), an \emph{endofunctor coalgebra} for \(F\) is a pair \((A, \alpha)\) where \(\alpha\) is a morphism of type \(A \to FA\) in \(\C\). Morphisms between coalgebras $(A,\alpha) \to (B,\beta)$ are morphisms $h\colon A \to B$ in \(\C\) such that the following square commutes.
\[
    \begin{tikzcd}
        A \rar{h}\dar[swap]{\alpha} & B \dar{\beta}
        \\
        FA \rar[swap]{F h} & FB
    \end{tikzcd}
\]

On the other hand, given a comonad \(\CC\) on \(\C\), an \emph{Eilenberg--Moore coalgebra} for~\(\CC\) is an endofunctor coalgebra \((A, \alpha)\), i.e.\ a morphism $\alpha\colon A \to \CC A$, such that the following diagrams commute.
\begin{equation*}
    \begin{tikzcd}
        A \dar[swap]{\alpha} \ar{dr}{\id} \\
        \CC A \rar[swap]{\epsilon_A} & A
    \end{tikzcd}
    \qquad
    \begin{tikzcd}
        A \rar{\alpha}\dar[swap]{\alpha} & \CC^2 A \dar{\delta_A}  \\
        \CC A \rar[swap]{\CC\alpha} & \CC^2 A
    \end{tikzcd}
\end{equation*}
We denote by $\EM(\CC)$ the category of Eilenberg--Moore coalgebras for \(\CC\) and coalgebra morphisms between them. The forgetful functor \(\EM(\CC) \to \C\), that sends \((A, \alpha)\) to \(A\), has a right adjoint that sends \(A\) to the cofree coalgebra \((\CC(A), \delta_A)\).
Most coalgebras of interest in this paper are Eilenberg--Moore coalgebras. If no confusion arises, we refer to them simply as \emph{coalgebras}.

\subsection{Factorisation systems}\label{ss:fact-system}

We review some basic facts concerning factorisation systems; more details can be found, e.g., in~\cite[Chapter 14]{adamek2004abstract} and \cite[\S11.2]{Riehl2014}.
A pair of classes of morphisms \((\Q,\M)\) on a category \(\C\) is an \emph{orthogonal factorisation system} if
\begin{enumerate}
    \item\label{ax:d1'} every morphism \(f\) in \(\C\) can be written as \(f = m\cdot q\), where \(q\in \Q\) and \(m\in \M\);
    \item\label{ax:o2} for every \(q\in \Q\) and \(m \in \M\), 
    and for every commutative square as on the left-hand side below,
\[
\begin{tikzcd}
    \bullet \rar{q} \dar & \bullet \dar
    \\
    \bullet \rar{m}      & \bullet
\end{tikzcd}
\qquad
\qquad
\begin{tikzcd}
    \bullet \rar{q} \dar & \bullet \dar \ar[swap]{ld}{d}
    \\
    \bullet \rar{m}      & \bullet
\end{tikzcd}
\]
there exists a unique \emph{diagonal filler}, i.e.\  an arrow \(d\) such that the right-hand diagram above commutes;

    \item\label{ax:i0+ci2} \(\Q\) and \(\M\) are closed under composition with isomorphisms.
\end{enumerate}

Let $(\Q,\M)$ be an orthogonal factorisation system in a category. We refer to $\Q$-morphisms as \emph{quotients} and denote them by $\epi$, and to $\M$-morphisms as \emph{embeddings} and denote them by~$\emb$.
We say that $(\Q,\M)$ is a \emph{proper factorisation system} if every quotient is an epimorphism and every embedding is a monomorphism. Moreover, we say that $(\Q,\M)$ is \emph{stable} if pullbacks of quotients along embeddings exist and are quotients.

We will use the following standard facts about proper factorisation systems \((\Q, \M)\) without further reference (see e.g.~\cite{freyd1972categories,Riehl2014}):
\begin{itemize}
    \item both \(\Q\) and \(\M\) are closed under compositions;
    \item \(\Q\cap \M\) is precisely the class of isomorphisms;
    \item the class \(\M\) is closed under all existing pullbacks along arbitrary morphisms;
    \item \(g\circ f \in \M\) implies \(f\in\M\) and, dually, \(g\circ f \in \Q\) implies \(g\in\Q\).
\end{itemize}

The following result allows us to lift factorisation systems from the base category to the category of coalgebras for a comonad. 

\begin{lemma}
    \label{l:lift-fs}
    Let \(\CC\) be a comonad on a category \(\C\) with a proper factorisation system \((\Q,\M)\). If \(\CC\) preserves embeddings, then \(\EM(\CC)\) admits a proper factorisation system $(\o{\Q}, \o{\M})$ where a morphism of coalgebras $h \colon (A,\alpha) \to (B,\beta)$ is in~$\o{\Q}$ (resp.~$\o{\M}$) if the underlying morphism $h \colon A \to B$ is in~$\Q$ (resp.~$\M$). If, in addition, \((\Q,\M)\) is stable and \(\CC\) sends pullbacks of quotients along embeddings to weak pullbacks, then $(\o{\Q}, \o{\M})$ is also stable.
\end{lemma}

\begin{prf}
The fact that $(\o{\Q}, \o{\M})$ is a proper factorisation system whenever \(\CC\) preserves embeddings follows from 
an adaptation of~\cite[Lemma 2]{linton1969coequalizers}.

For the second part of the statement, let \(e\colon (A,\alpha) \epi (B, \beta)\) and \(m\colon (S,\sigma) \emb (B, \beta)\) be, respectively, a quotient and an embedding in \(\EM(\CC)\). We denote the pullback of the underlying morphisms in \(\C\) as follows.
    \[
        \begin{tikzcd}
            T \rar[->>]{e'} \dar[>->,swap]{m'} & S \dar[>->]{m}
            \\
            A \rar[->>]{e} & B
        \end{tikzcd}
    \]
    As $\CC$ sends the latter pullback to a weak pullback, there exists a morphism \(\tau\) making the following diagram commute.
    \[
        \begin{tikzcd}
            \CC T \ar[swap]{ddd}{\CC m'}\ar{rrr}{\CC e'} & & & \CC S \ar{ddd}{\CC m}
            \\
            & T \ar[dashed]{lu}{\tau}\rar[->>]{e'} \dar[>->,swap]{m'} & S \dar[>->]{m} \ar[swap]{ru}{\sigma}
            \\
            & A \ar[swap]{ld}{\alpha} \rar[->>]{e} & B \ar{rd}{\beta}
            \\
            \CC A \ar{rrr}{\CC e} & & & \CC B
        \end{tikzcd}
    \]

We claim that \((T,\tau)\) is an Eilenberg--Moore coalgebra.
    By naturality of \(\epsilon\) we see that
    \[
        m' \circ \epsilon_T \circ \tau = \epsilon_A \circ \CC m' \circ \tau = \epsilon_A \circ \alpha \circ m' = m',
    \]
    where the last equality holds because \((A, \alpha)\) is a coalgebra. Hence, since \(m'\) is a monomorphism, we get that \(\epsilon_T \circ \tau = \id\).
Next, to check the equality \(\delta_T \circ \tau = \CC \tau \circ \tau\), we observe that
    \[
        \CC^2 m' \circ \delta_T \circ \tau
        = \delta_A \circ \CC m' \circ \tau
        = \delta_A \circ \alpha \circ m'
        = \CC \alpha \circ \alpha \circ m',
    \]
    where the first equality holds by naturality of \(\delta\), the second by the commutativity of the left face in the above diagram, and the third because \((A, \alpha)\) is a coalgebra. Then, using the commutativity of the left face twice, we have that
    \[
        \CC \alpha \circ \alpha \circ m' = \CC \alpha \circ \CC m' \circ \tau = \CC^2 m' \circ \CC \tau \circ \tau.
    \]
    Therefore, we get \(\CC^2 m' \circ \delta_T \circ \tau = \CC^2 m' \circ \CC \tau \circ \tau\), which gives us the desired equality since \(\CC^2 m'\) is a monomorphisms (just recall that \(\CC\) preserves embeddings).

    It remains to show that \((T,\tau)\) is a pullback of \(m\) and \(e\).
    Suppose that \(h\colon (Z, \zeta) \to (A,\alpha)\) and \(l\colon (Z, \zeta) \to (S, \sigma)\) satisfy \(f \circ h = g \circ l\).
    Since the underlying morphisms commute in \(\C\), there is a \(d\colon Z \to T\) such that \(h = m' \circ d\) and \(l = e' \circ d\).
    To see that \(d\) is a morphism \((Z, \zeta) \to (T, \tau)\), observe that
    \[
        \CC m' \circ \tau \circ d = \alpha \circ m' \circ d = \CC m' \circ \CC d \circ \zeta
    \]
    and, since \(\CC m'\) is a monomorphism, we get the required equality \(\tau \circ d = \CC d \circ \zeta\).
    Uniqueness of $d$ follows from the fact that \(m'\) is a monomorphism.
\end{prf}

\begin{remark}
Note that the definition of an orthogonal factorisation system varies across the literature, e.g.~\cite{Riehl2014} uses a different definition. See \cite[Chapter 14]{adamek2004abstract} and~\cite{carboni1997localization} for statements that imply the equivalence with the definitions in \emph{op.~cit.}
\end{remark}

\subsubsection{The posets of $\M$-subobjects.}
Let $\C$ be a well-powered\footnote{Recall that a category is \emph{well-powered} if the collection of subobjects of each object in the category forms a set rather than a proper class.} category endowed with a proper factorisation system $(\Q,\M)$. 
The poset $\Emb{X}$ of \emph{$\M$-subobjects} of an object $X\in\C$ is defined similarly to how one usually defines the poset of subobjects of \(X\).
Formally, we first define the preorder \(\leq\) on the class of embeddings \(S \emb X\) as follows. For embeddings \(m\colon S \emb X\) and \(n\colon R \emb X\), set
\[
    m \leq n
    \qtq{if}
    \exists h\colon S \to R
    \ete{such that}
    m = n \cdot h.
\]
(If it exists, $h$ is in $\M$.) Then, \(\Emb X\) is the poset reflection of the preorder \(\leq\) of embeddings \(S \emb X\), i.e.\ the poset of equivalence classes \([m] = \{ n\colon R \emb X \mid m \leq n \ete{and} n \leq m\}\) partially ordered by \([m] \leq [n]\) if, and only if, \(m \leq n\).

\section{Game comonads and their categories of coalgebras}
\label{s:game-comonads}

We review the three main game comonads, already mentioned in the \hyperref[s:intro]{Introduction}, from~\cite{AbramskyS18,AS2021}. Experience shows that the intuition gained from working with these usually transfers to other game comonads and arboreal categories.

\subsubsection{Relational structures.}
The three comonads are defined on the category \(\RS\) of \emph{relational structures} in a fixed relational signature \(\sg = \left<R_1, \dots, R_n\right>\) (or, in the case of the modal comonads, on pointed relational structures; see below for more details). For each relational symbol \(R\in \sg\), write \(\ar(R)\in \mathbb N\) for its arity. Then, the objects of \(\RS\) are tuples \(X = (X, R_1^X, \dots, R_n^X)\) where \(X\) is a set and \(R^X \sue X^{\ar(R)}\) for each \(R\in \sg\).
The morphisms of \(\RS\) are homomorphism of relational structures, i.e., functions \(h\colon X\to Y\) between the underlying sets such that, for each \(R\in \sg\),
\begin{equation}
    R^X(x_1,\dots, x_n)
    \qtq{implies}
    R^Y(h(x_1),\dots,h(x_n)).
    \label{eq:homo}
\end{equation}

Similarly, \(\pRS\) is the category of \emph{pointed relational structures}, i.e.\ structures \(X\) with a \emph{distinguished point} \(\point \in X\) which needs to be preserved by homomorphisms.
When \(\sg\)~is a \emph{modal signature}, i.e.\ each relation symbol in \(\sg\) is either unary (representing propositional predicates) or binary (representing transition relations), \(\pRS\) is the category of \emph{pointed Kripke frames}.
Note that Kripke homomorphisms are strictly more general than the usual \emph{p-morphisms} (also called \emph{bounded morphisms} \cite[p.~17]{BvBW2007}). This is necessary as otherwise the counits of our comonads would not be morphisms in the category.

Both \(\RS\) and \(\pRS\) are equipped with the (epi, regular mono) factorisation system. Concretely, this means that quotients are the surjective homomorphisms, and embeddings are the injective homomorphisms that reflect relations (meaning that \eqref{eq:homo} is an equivalence rather than a mere implication).

\paragraph{Coalgebraic view.}
Note that \(\RS\) can be viewed as a category of endofunctor coalgebras for the functor \(F\colon \Set \to \Set\) given by
\begin{equation}
    F(X) = \prod_{R\in \sg} \Pow(X^{\ar(R)-1}).
    \label{eq:signature-to-endofunc}
\end{equation}
However, morphisms in \(\RS\) are just the \emph{lax} coalgebra morphisms, depicted on the left-hand side below.
\begin{equation}
    \begin{tikzcd}
        X \rar{h}\dar[swap]{\xi}\ar[phantom]{rd}{\sue} & X' \dar{\xi'}
        \\
        FX \rar[swap]{F h} & FX'
    \end{tikzcd}
    \qquad
    \qquad
    \begin{tikzcd}[column sep=0.6em]
        & \point
        \ar{ld}
        \ar{rd}
        \\
        X \ar{rr}{h}\ar[swap]{d}{\xi}\ar[phantom]{rrd}{\sue} && X' \dar{\xi'}
        \\
        FX \ar[swap]{rr}{F h} && FX'
    \end{tikzcd}
    \label{eq:lax-morph}
\end{equation}
The ordering is the pointwise subset ordering, given by the power sets in~\eqref{eq:signature-to-endofunc}.
Similarly, when \(\sigma\) is a modal signature, \(\pRS\) corresponds to the category of pointed coalgebras and lax morphisms, as shown on the right-hand side above, for the functor \(F(X) = \Pow(A) \times \Pow(X)^T\) where \(A\) is the set of propositional predicates and \(T\) is the set of (binary) transition symbols.

Embeddings in the (epi, regular mono) factorisation system are the injective lax coalgebra morphisms, as in \eqref{eq:lax-morph}, where \(\xi\) is the largest possible.

\begin{remark}
Although our definitions assume that the signature \(\sg\) is finite and finitary, most results and notions apply equally well to infinite and infinitary signatures. However, subtle issues arise when discussing the relationship between game comonads and logic, and these are best avoided here. 
\end{remark}

\subsubsection{Modal comonad.}
The first comonad we define is the \emph{modal comonad} on pointed Kripke structures. More precisely, for each modal signature~\(\sg\) and each natural number \(k\geq 1\), we define a comonad \(\Mk\) on \(\pRS\) as follows. Given a pointed Kripke structure \(A\in \pRS\), \(\Mk A\) is the \emph{$k$-unravelling} of~\(A\), that is:
\begin{itemize}
    \item The underlying set of \(\Mk A\) consists of sequences
    \[
        \point \xrightarrow{R_{\tau_i}} a_1 \xrightarrow{R_{\tau_2}} a_2 \xrightarrow{R_{\tau_3}} \cdots \xrightarrow{R_{\tau_n}} a_n
    \]
    of length at most \(k\).

    We encode such a sequence as \(s=[a_0,\tau_1,a_1,\tau_2,a_2,\, \dots,\, \tau_n, a_n]\) where \(0\leq n \leq k \), \(a_0\) is the distinguished point \(\point \in A\) and, for every \(i\in \{1, \dots, n\}\), we have \(R_{\tau_i}^A(a_{i-1}, a_i)\). There is a function \(\epsilon_A\colon \Mk A \to A\) that extracts the last element of a sequence, i.e., it sends \([\point,\tau_1,a_1,\, \dots,\, \tau_n, a_n]\) to \(a_n\).

    \item The distinguished point of \(\Mk A\) is the shortest sequence, that is, \([\point]\).
    \item For a unary symbol \(P\in \sg\), define \(P^{\Mk A}\) as \(\{s \mid P^A(\epsilon_{A}(s))\}. \)
          For a binary symbol \(R_{\tau}\in \sg\), define \(R^{\Mk A}_\tau\) as the set of pairs \((s,s \pp [\tau,a])\), with \(a\in A\), such that $R^{A}_{\tau}(\epsilon_A(s),a)$. Here, \(\pp\) denotes the concatenation of sequences.
\end{itemize}
Note that \(\Mk\) is an endofunctor; its action on morphisms is given by
\[
    \Mk(f)\colon [\point,\tau_1,a_1, \, \dots,\, \tau_n, a_n] \mapsto [\point,\tau_1, f(a_1), \, \dots,\, \tau_n, f(a_n)].
\]
The comonad structure is defined as follows:
\begin{itemize}
    \item The counit is the map \(\epsilon_A\colon \Mk A \to A\) defined above.
    \item The comultiplication \(\delta_A\colon \Mk A \to \Mk^{2} A\) sends a sequence to the \emph{sequence of its prefixes}, i.e.\ \([\point,\tau_1,a_1,\, \dots,\, \tau_n, a_n]\) is mapped to
    \[[[\point],\tau_1,[\point, \tau_1, a_1], \tau_2,\, \dots,\, \tau_n, [\point,\tau_1,a_1,\, \dots,\, \tau_n, a_n]].\]
\end{itemize}

Observe that, in general, \(\epsilon_A\) is not a p-morphism since the length of sequences in \(\Mk A\) is bounded by \(k\). However, the structures \(A\) and \(\Mk A\) are indistinguishable in the fragment \(\ML_k\) of modal logic, consisting of formulas of \emph{modal depth} at most \(k\) (i.e., with at most $k$ nested modalities). In particular,
\begin{equation}
    A \equiv_{\ML_k} B
    \qq\iff
    \Mk A \sim \Mk B
    \label{eq:MLk-equiv}
\end{equation}
where \(\sim\) denotes bisimilarity or, equivalently, the existence of a winning strategy of Duplicator in the bisimulation game.

These facts are well known in the literature on modal logic, see e.g.~\cite{hennessy1980observing}, except perhaps for the observation that \(\Mk\) is a comonad. Interestingly, a similar pattern emerges in many other model-comparison games occurring in (finite) model theory.
In the case of other comonads, however, the category of coalgebras typically plays a more prominent role, as it must be explicitly invoked to express bisimilarity.

\subsubsection{\EF{} comonad.}
The next two comonads that we define are on \(\RS\), for a given relational signature \(\sg\).
The first comonad can be viewed as an unravelling of a structure based on the rules of the \EF{} game.
For a fixed \(k\) and a relational structure \(A\in \RS\), define \(\Ek A\) as follows:
\begin{itemize}
    \item The universe of \(\Ek A\) is the set of sequences \([a_1, \dots, a_n]\) where \(1\leq n\leq k\) and, for every \(i\in \{1,\dots, n\}\), \(a_i \in A\).
    \item \(\Ek(f)\colon \Ek A \to \Ek B\) maps \([a_1,\dots,a_n]\) to \([f(a_1),\dots,f(a_n)]\).
    \item \(\epsilon_A\colon \Ek A \to A\) maps \([a_1,\dots,a_n]\) to \(a_n\).
    \item \(\delta_A\colon \Ek A \to \Ek^2 A\)  maps a sequence to the sequence of its prefixes. 
    \item For an \(n\)-ary relation symbol \(R\in \sg\), the relation \(R^{\Ek A}\) consists of the tuples of sequences \((s_1, \dots, s_n)\) such that
    \begin{itemize}
        \item \(s_1,\dots,s_n\) are pairwise comparable in the prefix preorder, and
        \item \(R^A(\epsilon_A(s_1),\dots,\epsilon_A(s_n))\).
    \end{itemize}
\end{itemize}
Similarly to \eqref{eq:MLk-equiv}, we have that
\begin{equation}
    A \equiv_{\mathcal{L}_{\omega,k}^-} B
    \qq\iff
    \Ek A \sim \Ek B
    \label{eq:Ek-equiv}
\end{equation}
where \(\mathcal{L}_{\omega,k}^-\) denotes the fragment of first-order logic consisting of sentences without equality and with \emph{quantifier rank} at most \(k\) (i.e., with at most $k$ nested quantifiers).

This time, however, it is essential to interpret the right-hand side of \eqref{eq:Ek-equiv} in~\(\EM(\Ek)\). Namely, \(\Ek A\) and \(\Ek B\) are the cofree coalgebras on \(A\) and \(B\), respectively, and the relation \(\sim\) is defined internally in \(\EM(\Ek)\), which we discuss in Section~\ref{s:bisim} below. Also note that formulas with equality can be recovered if we work with the extended signature \(\sg \cup \{I\}\) where \(I\) is binary and \(I^A\) is interpreted as the diagonal \(\{ (a,a) \mid a\in A\}\), for details see \cite{AS2021}.

\subsubsection{Pebble comonad.}
The third family of game comonads is based on \emph{pebble games}.
For a fixed \(k\) and a relational structure \(A\in \RS\), define \(\Pk A\) as follows:
\begin{itemize}
    \item The universe of \(\Pk A\) is the set of sequences \([(p_1,a_1), \dots, (p_n,a_n)]\) where, for every \(i\in \{1,\dots, n\}\), \(p_i \in \{1,\dots, k\}\) and \(a_i \in A\).
    \item \(\Pk(f)\colon \Pk A \to \Pk B\) maps \([(p_1,a_1),\dots,(p_n,a_n)]\) to \([(p_1,f(a_1)),\dots,(p_n,f(a_n))]\).
    \item \(\epsilon_A\colon \Pk A \to A\) maps \([(p_1,a_1),\dots,(p_n,a_n)]\) to \(a_n\).
    \item \(\delta_A\colon \Pk A \to \Pk^2 A\) maps a sequence to the sequence of its prefixes.
    \item For an \(n\)-ary relation symbol \(R\in \sg\), the relation \(R^{\Pk A}\) consists of the tuples of sequences \((s_1, \dots, s_n)\) such that
    \begin{itemize}
        \item \(s_1,\dots,s_n\) are pairwise comparable in the prefix preorder;
        \item \(R^A(\epsilon_A(s_1),\dots,\epsilon_A(s_n))\);
        \item if \(s_i\) is a prefix of \(s_j\), i.e.\ \(s_i = [(p_1,a_1), \dots, (p_n,a_n)]\) and \(s_j = s_i \pp [(q_1,a'_1), \dots, (q_r, a'_r)]\), then \(p_n \not\in \{q_1, \dots, q_r\}\).
    \end{itemize}
\end{itemize}
Again, similarly to \eqref{eq:MLk-equiv}, we have that
\begin{equation}
    A \equiv_{\mathcal{L}_{\infty,\omega}^{-,k}} B
    \qq\iff
    \Pk A \sim \Pk B
    \label{eq:Pk-equiv}
\end{equation}
where \(\mathcal{L}_{\infty,\omega}^{-,k}\) denotes the fragment of infinitary first-order logic $\mathcal{L}_{\infty,\omega}$ consisting of sentences without equality that use at most \(k\) distinct variables. This fragment is pivotal in finite model theory~\cite{Libkin2004,GradelKLMSVVW07}.
For example, the existence of a path of length 4 in a graph can be expressed by the following sentence in \(\mathcal{L}_{\infty,\omega}^{-,2}\):
\[\exists x_1 (\exists x_2. R(x_1, x_2) \land (\exists x_1. R(x_2, x_1) \land (\exists x_2. R(x_1, x_2) \land (\exists x_1. R(x_2, x_1))))).\]

\subsection{The categories of coalgebras}
\label{s:coalg-concretely}
For the following, it is useful to recall from~\cite{AS2021} the concrete description of the categories of coalgebras of our running examples of game comonads.
To this end, we say that a poset \((F, \leq)\) is a \emph{forest} if, for every \(x\in F\), the set 
\[\down\, x \coloneqq \{y \in F\mid y\leq x\}\] is a finite chain. A function between forests \(f\colon F\to F'\) is a \emph{forest morphism} if it is monotone and the chains \(\down\, x\) and \(\down\, f(x)\) have the same length for each \(x\in F\).

If $x$ is an element of a forest $F$, its \emph{height}, denoted by $\htf(x)$, is the cardinality of the set $\down \, x \setminus \{x\}$.
The elements of height $0$ are called \emph{roots}.
A \emph{tree} is a forest with precisely one root.
The \emph{height of the forest} $F$ is the supremum of the set \(\{\htf(x) \mid x\in F\}\) in \(\mathbb{N}\cup\{\infty\}\).

The coalgebras for our three game comonads have an inherent forest order. Namely, for a comonad \(\CC\) among \(\Mk, \Ek, \Pk\) and a coalgebra \(\alpha\colon A \to \CC A\), set
\[
    a \sqleq_\alpha  b
    \qq\iff
    \ete{the sequence} \alpha(a) \ete{is a prefix of} \alpha(b).
\]
Then, the poset \((A, \sqleq_\alpha)\) is a forest and coalgebra morphisms are forest morphisms.
This allows us to give a concrete characterisation of the category \(\EM(\CC)\) as a class of forest-ordered structures:
\begin{itemize}
    \item Since \(\Mk\) is an idempotent comonad, the category \(\EM(\Mk)\) is isomorphic to the full subcategory of \(\pRS\) consisting of \emph{synchronisation trees}~\cite{JNW1996} of height \({\leq}\, k\), i.e.\ tree-ordered $\sg$-structures \((A, \leq)\) such that \(a\) is the parent of \(b\) just when \(R^A(a,b)\) for some relation \(R\) in \(\sg\).
In particular, Kripke morphisms between synchronisation trees are automatically forest morphisms.

    \item The objects of \(\EM(\Ek)\) can be described as forest-ordered $\sg$-structures \({(A,\leq)}\) of height \({\leq}\, k\) such that if \(R^A(x_1,\dots, x_n)\), then \(x_1, \dots, x_n\) are pairwise comparable in the order \(\leq\).
    The latter condition follows from the definition of the relation \(R^{\Ek A}\) given above.

    Coalgebra morphisms correspond precisely to homomorphisms of the underlying relational structures that are also forest morphisms.

    \item Objects in \(\EM(\Pk)\) can be described as forest-ordered $\sg$-structures \({(A, \leq, p)}\) with a ``pebbling function'' \(p\colon A \to \{1, \dots, k\}\) such that \(R^A(x_1,\dots, x_n)\)
    implies that (i) \(x_1, \dots, x_n\) are pairwise comparable and (ii) if \(x_i < x_j\) then \(p(x_i) \not\in \{ p(y) \mid x_i < y \leq x_j\}\).
    Morphisms in \(\EM(\Pk)\) are the homomorphisms of relational structures that preserve the pebbling functions and are also forest morphisms.
\end{itemize}

These categories of coalgebras are very well-behaved. For instance, they are complete and cocomplete, admit an (epi, regular mono) factorisation system, and are even locally finitely presentable. Cf.~\cite{DJR2021,jakl2023categorical,Reggio2021polyadic}.

\subsection{Game comonads and logic}
\label{s:bisim}

Using the concrete characterisation of the category \(\EM(\CC)\), with \(\CC\) one of the game comonads $\Mk$, $\Ek$, or $\Pk$, we can express notions familiar from the universal coalgebra literature.
To this end, we say that a coalgebra \((A,\alpha)\) is a \emph{path} if it is a finite chain in the induced order \(\sqleq_\alpha\).
Then, following~\cite{JNW1996}, we say that a morphism \(f\colon X \to Y\) in \(\EM(\CC)\) is \emph{open} if any square
\[
    \begin{tikzcd}
        P\dar\rar & R \dar
        \\
        X \rar{f} & Y
    \end{tikzcd}
\]
with \(P, R\) paths has a diagonal filler \(R \to X\), making the ensuing diagram commute. Open morphisms can be regarded as an abstraction of p-morphisms.

Using these notions, we can introduce three types of equivalence in~\(\EM(\CC)\).
For coalgebras \(X\) and \(Y\), we say that they are
\begin{itemize}
    \item \emph{back-and-forth equivalent} if there is a bisimulation \(B \sue X\times Y\) (in the sense of \cite[p.~14]{BvBW2007}) with respect to the Kripke structures \((X,\prec)\) and \((Y,\prec)\), where \(\prec\) denotes the immediate-successor relation in the forest order, and for each \mbox{\((x,y)\in B\)}, the induced paths \(\down\, x\sue X\) and \(\down\, y\sue Y\) are isomorphic;
    \item \emph{bisimilar} if there is a span of open morphisms \(X \leftarrow Z \to Y\).
\end{itemize}

If the relations \(\sim\) appearing in \eqref{eq:MLk-equiv}, \eqref{eq:Ek-equiv} and \eqref{eq:Pk-equiv} are set to be the back-and-forth equivalence, then the latter equations express the well known fact that the existence of a Duplicator winning strategy in the \(k\)-round bisimulation, \(k\)-round \EF{} and \(k\)-pebble games, respectively, corresponds precisely to the logical equivalence on the left-hand side of these equations.

Furthermore, in the setting of arboreal categories, which we introduce below, back-and-forth equivalence and bisimilarity coincide, as shown in \cite{AR2024}.

\section{Arboreal categories and the failure of connectedness}
\label{s:arboreal-old}

The \emph{raison d'être} of arboreal categories is to provide a general language for working with back-and-forth equivalence, bisimilarity, and so forth, for categories of coalgebras for game comonads. This is achieved uniformly, without relying on any specific feature of the comonads or the underlying category of structures.

An essential ingredient in the definition of arboreal categories is an axiomatic notion of path, which extends the concrete notion of path considered in Section~\ref{s:bisim} to any category with a proper factorisation system.

\begin{dfn}\label{def:paths}
Let \(\C\) be a well-powered category equipped with a proper factorisation system \((\Q,\M)\).
    An object $X$ of $\C$ is a \emph{path} if its poset of \(\M\)-subobjects~$\Emb{X}$ is a finite linear order. If $P$ is a path, its \emph{height}, denoted by $\htf(P)$, is the height of the finite chain \(\Emb P\).
    A \emph{path embedding} is an embedding $P\emb X$ whose domain is a path. 
\end{dfn}

The following is the definition of an arboreal category as given in~\cite{AR2021icalp}.

\begin{definition}[``Old definition'' of arboreal category]\label{def:arboreal-cat}
    Let \(\A\) be a well-powered category equipped with a stable proper factorisation system.
    We say that \(\A\) is \emph{arboreal} if it satisfies the following conditions:
\begin{description}\itemsep4pt
\item[Paths are connected]\label{ax:connected} Coproducts of sets of paths exist in $\A$ and each path $P$ is \emph{connected}; that is, every arrow $P\to \coprod_{i\in I}{Q_i}$ into a coproduct of a non-empty set of paths $\{Q_{i}\mid i\in I\}$ factors through some coproduct arrow $Q_{j}\to \coprod_{i\in I}{Q_i}$. 

\item[2-out-of-3 property] Given arrows $f\colon P\to Q$ and $g\colon Q\to R$ between paths, if $g\circ f$ is a quotient, then so is $f$.\footnote{This condition is equivalent to saying that if any two of $f$, $g$ and $g\circ f$ are quotients, then so is the third; see \cite[Remark~3.9]{AR2022}. Hence the term ``2-out-of-3 property''.}

\item[Path-generation]\label{i:path-gen}
    The full subcategory of paths is dense in \(\A\).\footnote{
        A category equipped with a stable proper factorisation system satisfies the path-generation property if, and only if, each object is \emph{path-generated}; that is, for each object~$X$, the cocone consisting of all path embeddings $P \emb X$ is a colimit cocone
        \cite[Lemma~5.1 and Remark~5.2]{AR2022}.
    }
\end{description}
\end{definition}

The paper \cite{AR2022}, in which arboreal categories were introduced, laid out the foundations for an axiomatic approach to logical equivalences in finite model theory. In fact, even game comonads introduced later still fit within the scope of arboreal categories, highlighting the robustness of this concept.
Consequently, results proved in the general setting of arboreal categories (cf.\ e.g.~\cite{ALR2025,AMS2024,AR2024,jakl2023categorical,reggio2023finitely}) also apply to these newly introduced comonads.

However, there is an issue. The categories of coalgebras for game comonads defined over \emph{pointed} Kripke structures may fail to be arboreal. This mistake has gone unnoticed until now. Below, we demonstrate that the connectedness axiom generally fails for modal comonads.

Firstly, we apply Lemma~\ref{l:lift-fs} to show that the (epi, regular mono) factorisation system of \(\pRS\) induces a stable proper factorisation system on the category $\EM(\Mk)$. It is straightforward to check that $\Mk$ preserves regular monomorphisms, so it remains to show that it sends pullbacks of quotients along embeddings to weak pullbacks.

\begin{lemma}
    \(\Mk\) sends pullbacks to weak pullbacks.
\end{lemma}
\begin{prf}
    Let \(f\colon A \to B\) and \(g\colon C \to B\) be homomorphisms in \(\pRS\), and let their pullback be as shown below.
    \[
        \begin{tikzcd}
            D \rar{f'} \dar[swap]{g'} & C \dar{g}
            \\
            A \rar{f} & B
        \end{tikzcd}
    \]
    We assume without loss of generality that \(D\) is the substructure of \(A \times C\) consisting of the pairs \((a,c)\) such that \(f(a) = g(c)\).

    Suppose there is a pointed Kripke structure \(Z\), and homomorphisms \(l\colon Z \to \Mk C\) and \(h\colon Z\to \Mk A\), such that the outer diagram below commutes.
    \[
        \begin{tikzcd}
            Z \ar[bend left=45]{rrd}{l}
              \ar[bend right=45,swap]{ddr}{h}
              \ar[dashed]{rd}{d}
            \\
            & \Mk D \rar{\Mk f'} \dar[swap]{\Mk g'} & \Mk C \dar{\Mk g}
            \\
            & \Mk A \rar[swap]{\Mk f} & \Mk B
        \end{tikzcd}
    \]
    We define a map \(d\colon Z\to \Mk D\). For any \(z\in Z\), we have that
    \[
        h(z) = [\bullet, \tau_1, a_1, \dots, \tau_n, a_n]
        \qtq{and}
        l(z) = [\bullet, \phi_1, c_1, \dots, \phi_m, c_m].
    \]
    Since the outer diagram above commutes, we get \(\Mk f(h(z)) = \Mk g(l(z))\). So, \(m = n\) and, for every \(i\in \{1, \dots, n\}\), \(\tau_i = \phi_i\) and \(f(a_i) = g(c_i)\).
    Therefore, we can define \(d(z) \coloneqq [\bullet, \tau_1, (a_1, c_1), \dots, \tau_n, (a_n, c_n)]\).
    It is immediate from the definition of \(d\) that the above diagram commutes.
    Finally, since both \(h\) and~\(l\) are homomorphisms, so is \(d\).
\end{prf}

\begin{remark}
The above argument can be adapted to show that the factorisation systems on $\EM(\Ek)$ and $\EM(\Pk)$ induced by the (epi, regular mono) factorisation system on $\RS$ are also stable proper factorisation systems.
\end{remark}

The paths in $\EM(\Mk)$ are precisely the synchronisation trees consisting of a single branch.
The next example shows that paths in $\EM(\Mk)$ need not be connected. Therefore, the category $\EM(\Mk)$ generally fails to be arboreal according to Definition~\ref{def:arboreal-cat}.

\begin{example}\label{ex:failure-connectedness}
Suppose the signature $\sigma$ contains two distinct unary relation symbols $r,s$. We define $\sigma$-structures $P_{r}, P_{s}, P_{r,s}$ based on a one-element universe: 
\begin{itemize}
\item $P_{r}$ consists of an element that satisfies $r$ but not $s$;
\item $P_{s}$ consists of an element that satisfies $s$ but not $r$;
\item $P_{r,s}$ consists of an element that satisfies both $r$ and $s$.
\end{itemize}
The $\sigma$-structures $P_{r}, P_{s}, P_{r,s}$, with their unique element as the distinguished element, are paths in $\EM(\Mk)$ and satisfy $P_{r,s} \cong P_r + P_s$ in $\EM(\Mk)$. However, the isomorphism $P_{r,s}\to P_r + P_s$ does not factor via either of the inclusions $P_r \to P_r + P_s$ or $P_s \to P_r + P_s$. Hence, $P_{r,s}$ is a path that is not connected according to Definition~\ref{def:arboreal-cat}.
\end{example}

\begin{remark}
 The failure of connectedness of paths in Example~\ref{ex:failure-connectedness} is related to the presence of a distinguished point, and does not apply to the \EF{} or pebble comonads. In particular, \(\EM(\Ek)\) and \(\EM(\Pk)\) are arboreal according to Definition~\ref{def:arboreal-cat}.
    On the other hand, we expect that connectedness of paths will also fail for other game comonads defined on pointed relational structures, such as the hybrid \cite{AM2022}, PPML~\cite{FigueiraGoren2025} and pebble-relation~\cite{MS2022} comonads.
\end{remark}

\section{Arboreal categories redefined: tree-connectedness}\label{s:arboreal-new}

We propose amending the definition of arboreal category by replacing the connectedness condition with a \emph{tree-connectedness} condition. The 2-out-of-3 and path-generation axioms remain unchanged.

\begin{definition}
Let $\C$ be a category equipped with a stable proper factorisation system. A \emph{tree-diagram} in $\C$ is a functor $D\colon T\to \C$ where $T$ is a tree order, and the image of $D$ consists of embeddings in $\C$. We say that $D$ is a \emph{tree-diagram of paths} if, in addition, its image consists of embeddings between paths.
\end{definition}

\begin{definition}[``New definition'' of arboreal category]\label{def:arboreal-cat-new}
    Let \(\A\) be a well-powered category equipped with a stable proper factorisation system.
    We say that \(\A\) is \emph{arboreal} if it satisfies the following conditions:
\begin{description}\itemsep4pt
\item[Paths are tree-connected]\label{ax:tree-connected} $\A$ admits an initial object and colimits of tree-diagrams of paths, and each path $P$ is \emph{tree-connected}; i.e., for every tree-diagram of paths ${D\colon T\to \A}$ and every arrow $f\colon P\to \colim{D}$, there is a least $t\in T$ such that $f$ factors through the colimit arrow $D(t)\to \colim{D}$. 

\item[2-out-of-3 property] Given arrows $f\colon P\to Q$ and $g\colon Q\to R$ between paths, if $g\circ f$ is a quotient, then so is $f$.

\item[Path-generation] The full subcategory of paths is dense in \(\A\).
\end{description}
\end{definition}

The diagram of paths in Example~\ref{ex:failure-connectedness}, which showed the failure of connectedness of paths in the category $\EM(\Mk)$, is not a tree-diagram of paths and cannot be turned into one.
In fact, we have the following:

\begin{lemma}
Every path in the category $\EM(\Mk)$ is tree-connected.
\end{lemma}

\begin{prf}
Let $D\colon T\to \EM(\Mk)$ be a tree-diagram of paths in $\EM(\Mk)$, with $T$ a tree. First, note that for every $u\in T$ the colimit map $i_{u}\colon D(u)\to \colim{D}$ is an embedding. It is injective because it is a forest morphism from a chain. To see that it reflects unary relations, observe that the full inclusion $\EM(\Mk) \to \mathcal \pRS$ is left adjoint and so it preserves colimits. This means that $\colim{D}$ is computed by first taking the coproduct of the $\sigma$-structures $D(u)$, for $u\in T$, which is given by glueing the structures $D(u)$ at their distinguished elements, and then quotienting according to the embeddings in the diagram. Hence, if $x\in D(u)$ and $i_{u}(x)$ satisfies a relation $r$, there is $x'\in D(u')$ that satisfies $r$ and such that $i_{u}(x)=i_{u'}(x')$. But this can only happen if there are $t\leq u,u'$ and $y\in D(t)$ such that $D(t\leq u)(y)=x$ and $D(t\leq u')(y)=x'$. Since $D(t\leq u)$ and $D(t\leq u')$ are embeddings, $x$ satisfies $r$ because $x'$ does.

Now, consider an arrow $f\colon P\to \colim{D}$ with $P$ a path, and decompose it as a quotient followed by an embedding:
\[\begin{tikzcd}
P \arrow[twoheadrightarrow]{r}{e} & Q \arrow[rightarrowtail]{r}{m} & \colim{D} 
\end{tikzcd}\]
As $P$ is a path, i.e.\ it consists of a single branch of height at most $k$, so is $Q$. The image of the homomorphism $m$ is of the form $\down\, x = \{y\in \colim{D}\mid y\leq x\}$ for a unique element $x\in \colim{D}$. Since every arrow in the image of $D$ is an embedding, there is $t_{0}\in T$ such that the image of the colimit map $D(t_{0})\emb \colim{D}$ contains $\down \, x$ as a substructure. Thus, $m$ factors through the embedding $D(t_{0})\emb \colim{D}$. Because $T$ is a meet-semilattice in which every element has finite height, there is a least $t\in T$ such that $m$ factors through the colimit map $D(t)\emb \colim{D}$. It follows that $t$ is also the least element of $T$ such that $f$ factors through $D(t)\emb \colim{D}$.
\end{prf}

A similar argument shows that paths in the categories of coalgebras for the pebbling and Ehrenfeucht--Fra\"iss\'e comonads are tree-connected. Since the categories $\EM(\Mk)$, $\EM(\Ek)$ and $\EM(\Pk)$ are cocomplete (in particular, they admit colimits of tree-diagrams of paths), we conclude that:

\begin{proposition}
The categories $\EM(\Mk)$, $\EM(\Ek)$ and $\EM(\Pk)$ are arboreal in the sense of Definition~\ref{def:arboreal-cat-new}.
\end{proposition}

Regarding the relationship between connectedness and tree-connectedness, the following fact holds:

\begin{lemma}\label{l:old-vs-new-connectedness}
Let $\C$ be a category equipped with a stable proper factorisation system and admitting an initial object $\zero$. The following statements are equivalent:
\begin{enumerate}
\item\label{i:zero-no-proper-quo} $\zero$ has no proper quotients, i.e.\ any quotient with domain $\zero$ is an isomorphism;
\item\label{i:unique-arrow-emb} for every object $X\in \C$, the unique arrow $\zero\to X$ is an embedding.
\end{enumerate}
If either (and thus both) of these conditions are satisfied, and $\C$ is cocomplete, then any tree-connected object of $\C$ is connected.
\end{lemma}
\begin{prf}
    The equivalence of items~\ref{i:zero-no-proper-quo} and~\ref{i:unique-arrow-emb} is easily verified from basic properties of orthogonal factorisation systems. Now, let $X$ be a tree-connected object in $\C$ and let $\{Q_{i}\mid i\in I\}$ be a non-empty set of paths. Consider a tree $T$ consisting of a root with $I$-many immediate successors, and let $D\colon T\to \C$ be the functor that sends the root of $T$ to $\zero$, and the $i$-th immediate successor of the root to $Q_{i}$. This is a tree-diagram of paths whose colimit coincides with $\coprod_{i\in I}{Q_{i}}$. Since $X$ is tree-connected, any arrow $X\to \coprod_{i\in I}{Q_{i}}$ factors through some colimit arrow $Q_{j}\to \coprod_{i\in I}{Q_{i}}$.
\end{prf}

The following is an immediate consequence of Lemma~\ref{l:old-vs-new-connectedness}.

\begin{proposition}
Let $\C$ be an arboreal category in the new sense (Definition~\ref{def:arboreal-cat-new}). If $\C$ admits coproducts of sets of paths, and its initial object has no proper quotients, then $\C$ is arboreal in the old sense (Definition~\ref{def:arboreal-cat}).
\end{proposition}

Since both in \(\EM(\Ek)\) and \(\EM(\Pk)\) the initial object has no proper quotients, the above proposition provides another, less direct, proof that both categories are arboreal in the old sense. However, the proposition does not apply to \(\EM(\Mk)\). Note that as soon as the modal signature $\sigma$ contains a unary symbol $P$, any structure consisting of one element that satisfies $P$ is a proper quotient of the initial object, which is a one-element structure with no relations.

\section{Posets of $\M$-subobjects revisited}\label{s:posets-subobjects}

The main technical use of the connectedness axiom in~\cite{AR2021icalp,AR2022} was to establish properties of the posets of $\M$-subobjects. We show how to recover these properties under the new definition of an arboreal category. Henceforth, by ``arboreal category'' we mean a category that satisfies the conditions in Definition~\ref{def:arboreal-cat-new}. 

Reasoning as in~\cite[\S3]{AR2022}, we can associate with each object $X$ of an arboreal category the set $\Path{X}$ of (equivalence classes of) path embeddings into $X$. We regard the latter as a poset, with the order induced by that of $\Emb{X}$ (cf.~Section~\ref{ss:fact-system}). Moreover, every arrow $f\colon X\to Y$ induces a monotone map $\Path{f}\colon \Path{X}\to \Path{Y}$ that sends $m\colon P\emb X$ to the path embedding $Q\emb Y$ obtained by taking the $(\Q,\M)$-decomposition of $f\circ m$:
\[\begin{tikzcd}
P \arrow[twoheadrightarrow]{r}{} \arrow[rr, relay arrow=2ex, "f\circ m"] & Q \arrow[rightarrowtail]{r}{} & Y
\end{tikzcd}\] 
The object $Q$ is a path by \cite[Lemma~3.5]{AR2022}.

\begin{theorem}\label{thm:path-cat-functor-into-trees}
Let $\A$ be an arboreal category. The assignment $X\mapsto \Path{X}$ induces a functor $\Path\colon \A\to\T$ into the category $\T$ of trees and forest morphisms.
\end{theorem}
\begin{prf}
For arboreal categories in the old sense, this is~\cite[Theorem~3.11]{AR2022}. Its proof only requires a category with an initial object and a stable proper factorisation system satisfying the 2-out-of-3 condition~\cite[Remark~3.14]{AR2022}.
\end{prf}

\begin{lemma}\label{l:path-cat-suprema}
For any object $X$ of an arboreal category~$\A$, the following hold:
\begin{enumerate}[label=(\alph*)]
\item\label{suprema-of-paths} Any subset $\U\subseteq \Path{X}$ admits a supremum $\bigvee \U$ in $\Emb{X}$.
\item\label{paths-j-irred} For any path embedding $[m]\in\Path{X}$ and non-empty set $\mathcal{S}\subseteq \Emb{X}$, if $[m]= \bigvee\mathcal{S}$ then $[m]\in\mathcal{S}$.
\item\label{at-most-one-emb} If \(m,n\colon P \emb Q\) are embeddings between paths, then \([m] = [n]\) in \(\Emb Q\).
\end{enumerate}
\end{lemma}

\begin{prf}
For arboreal categories in the old sense, this is~\cite[Lemma~3.15]{AR2022}.
Connectedness of paths plays no role in the original proof, but we need to adapt it to use colimits of tree-diagrams instead of coproducts of paths.

For item~\ref{suprema-of-paths}, consider a set of (equivalence classes of) path embeddings \mbox{\(\U \subseteq \Path{X}\)}.
The downward closure $I \coloneqq \down \, \U$ of $\U$ in $\Path{X}$ is a tree, and therefore we can define a tree-diagram of paths
\[
    D\colon I \to \A
\]
by sending an equivalence class \([m]\) to the domain \(P\) of a chosen representative \(m\colon P \emb X\) and \([m] \leq [m']\) to the unique \(e\colon P \emb P'\) such that \(m = m' \circ e\).
There is a compatible cocone on $D$ with vertex~$X$ whose component $D([m]) \to X$ is given by the chosen representative.
Let $S\coloneqq \colim{D}$ and consider the $(\Q,\M)$ factorisation of the unique mediating morphism $\delta\colon S\to X$:
\[\begin{tikzcd}
S \arrow[twoheadrightarrow]{r}{e} \arrow[rr, relay arrow=2ex, "\delta"] & T \arrow[rightarrowtail]{r}{n} & X
\end{tikzcd}\] 
Each representative \(m\) of $[m]\in \U$ factors through $n$, thus $[n]$ is an upper bound for~$\U$. We claim that $[n]$ is the least upper bound, i.e., $[n]=\bigvee{\U}$ in $\Emb{X}$. Suppose that all path embeddings in $\U$ factor through an embedding $n'\colon T'\emb X$. Then the same holds for all path embeddings in $\down \, \U$. By the universal property of $S$, we get a morphism $\phi\colon S\to T'$. By the uniqueness of $\delta$, we obtain $n'\circ \phi=\delta$, and so the following square commutes.
\[\begin{tikzcd}
S \arrow{d}[swap]{\phi} \arrow[twoheadrightarrow]{r}{e} & T \arrow[rightarrowtail]{d}{n} \\
T' \arrow[rightarrowtail]{r}{n'} & X
\end{tikzcd}\]
Therefore, there exists a diagonal filler $T\to T'$. In particular, the commutativity of the lower triangle entails that $n\leq n'$, as was to be proved.

The proofs of items~\ref{paths-j-irred} and~\ref{at-most-one-emb} are the same as in~\cite[Lemma~3.15]{AR2022}.
\end{prf}

For the next lemma, recall from Definition~\ref{def:paths} that $\htf(P)$ denotes the height of the path $P$.

\begin{lemma}\label{l:arboreal-consequences}
Let $\A$ be an arboreal category. The following statements hold:
\begin{enumerate}[label=(\alph*)]
\item\label{canonical-suprema} For any object $X$ of $\A$ and any $m\in \Emb{X}$, $m=\bigvee{\{p\in\Path{X}\mid p\leq m\}}$. 
\item\label{quot-to-surj} A morphism $f$ is a quotient if, and only if, $\Path{f}$ is surjective.
\item\label{i:height-quotient} A morphism $P\to Q$ between paths is a quotient if, and only if, \({\htf(P)=\htf(Q)}\).
\end{enumerate} 
\end{lemma}
\begin{prf}
For arboreal categories in the old sense, this is~\cite[Lemma~5.5]{AR2022} combined with \cite[Lemma~6.11]{ALR2025}. The original proofs rely on the existence of suprema of path embeddings into a given object, but not on their explicit construction via coproducts. Thus, the same proofs apply.
\end{prf}

\begin{proposition}\label{pr:perfect-lattice-of-strong-subs}
Let $\A$ be an arboreal category, $X$ an object of $\A$, and $\U\subseteq \Path{X}$ a non-empty subset. A path embedding $m\in\Path{X}$ is below $\bigvee{\U}$ if, and only if, it is below some element of~$\U$.
\end{proposition}
\begin{prf}
For arboreal categories in the old sense, this is~\cite[Proposition~5.6]{AR2022}.
We adapt the proof of the latter result.
Fix an object $X$ of~$\A$ and a non-empty set of path embeddings $\U=\{m_i\colon P_i\emb X\mid i\in I\}$. Let $m\colon P\emb X$ be an arbitrary path embedding. If $m$ is below some element of $\U$, then clearly $m\leq \bigvee\U$. 

For the converse direction, suppose $m\leq \bigvee\U$. Recall from the proof of Lemma~\ref{l:path-cat-suprema}\ref{suprema-of-paths} that the supremum of $\U$ is obtained by taking the $(\Q,\M)$ factorisation $\colim{D}\xepi{e} S\xemb{n} X$ of the mediating morphism $\colim{D}\to X$, where $D\colon T\to \A$ is the tree-diagram of paths associated with $\down \, \U$. With this notation, $\bigvee \U=n$. Since $m\leq \bigvee\U$, there exists an embedding $m'\colon P\emb S$ such that $m=n\circ m'$. Consider the pullback of $m'$ along $e$:
\[\begin{tikzcd}
V \arrow[rightarrowtail]{d}[swap]{j} \arrow[twoheadrightarrow]{r}{r} \arrow[dr, phantom, "\lrcorner", very near start] & P \arrow[rightarrowtail]{d}{m'} \\
 \colim{D} \arrow[twoheadrightarrow]{r}{e} & S
\end{tikzcd}\]
Applying Lemma~\ref{l:arboreal-consequences}\ref{quot-to-surj} to the quotient $r$, we see that there exists a path embedding $k\colon Q\emb V$ such that $\Path{r}(k)=\id_P$, i.e.~$r\circ k$ is a quotient. Because $Q$ is tree-connected, there is a least $t\in T$ such that $j\circ k\colon Q\emb \colim{D}$ factors through the colimit arrow $i_{t}\colon D(t)\to \colim{D}$, i.e., $j\circ k= i_{t}\circ p$ for some embedding $p\colon Q\emb D(t)$. Let $i\in I$ be such that the path embedding into $X$ corresponding to $D(t)$ is of the form $m_{i}\circ \lambda$ for some embedding $\lambda\colon D(t)\emb P_{i}$. We then have a commutative diagram as follows.
\[\begin{tikzcd}[column sep=3em, row sep=2.5em]
Q \arrow[twoheadrightarrow]{r}{r\circ k} \arrow[rightarrowtail]{d}[swap]{p} & P \arrow[rightarrowtail]{d}[swap]{m'} \arrow[rightarrowtail]{dr}{m} & {} \\
D(t) \arrow{r}{e\circ i_{t}} \arrow[rr, relay arrow=-2ex, "m_i\circ \lambda", swap, rightarrowtail] & S \arrow[rightarrowtail]{r}{n} & X
\end{tikzcd}\]
As $m\circ r\circ k=m_i\circ \lambda \circ p$ and the right-hand side of the equation is an embedding, $r\circ k$ is an isomorphism. So $m\leq m_{i}\circ \lambda \leq m_i\in \U$, thus concluding the proof.
\end{prf}

\begin{remark}
Combining Lemmas~\ref{l:path-cat-suprema}\ref{suprema-of-paths} and~\ref{l:arboreal-consequences}\ref{canonical-suprema} implies that, for any object $X$ in an arboreal category, its poset of embeddings $\Emb{X}$ admits all suprema and is therefore a complete lattice. 
\end{remark}

\begin{remark}
We could weaken the tree-connectedness condition in Definition~\ref{def:arboreal-cat-new} by considering only tree-diagrams of paths $D\colon T\to \A$ that admit a compatible cocone of embeddings in $\A$. This weaker axiom would still enable us to prove all the results in this section, since we only need to consider suprema of path embeddings into a fixed object. Similarly, it would suffice to assume that $\A$ has colimits of tree-diagrams of paths admitting a compatible cocone of embeddings. However, our main examples of arboreal categories are cocomplete.
\end{remark}

\section{The path functor is a Street fibration}
\label{s:path-func}

In this final section, we demonstrate an important consequence of the new definition of an arboreal category.
Let us fix an arboreal category $\A$ and write
\[
    \Path\colon \A\to\T
\]
for the associated \emph{path functor} into the category $\T$ of trees and forest morphisms between them (cf.\ Theorem~\ref{thm:path-cat-functor-into-trees}). Recall that $\Path$ assigns to an object $X$ of $\A$ the tree $\Path{X}$ of (equivalence classes of) path embeddings into $X$; the order on \(\Path X\) is the one induced by the poset \(\Emb X\) of $\M$-subobjects of $X$.

We wish to show that \(\Path\) is a Street fibration (Theorem~\ref{t:P-fibration} below).
To this end, we review some basic definitions related to fibrations.

\begin{dfn}
A morphism $f\colon x\to y$ in $\A$ is \emph{Cartesian (with respect to $\Path$)} if for all $g\colon z\to y$ in $\A$ and all $w\colon \Path z\to \Path x$ in $\T$ such that $\Path f\circ w = \Path g$, there exists a unique $\o{w}\colon z\to x$ in $\A$ such that $g = f\circ \o{w}$ and $\Path \o{w}=w$. 

\vspace{1em}
\begin{center}
\begin{tikzcd}[column sep = 4em]
z \arrow[bend left=30]{drr}[description]{g} \arrow[dashed]{dr}[description]{\exists! \, \o{w}} & & \\
 & x \arrow{r}[description]{f} & y
\end{tikzcd}

\vspace{0.5em}
\begin{tikzcd}[column sep = 3em]
\Path z \arrow[bend left=30]{drr}[description]{\Path g} \arrow{dr}[description]{w}& & \\
 & \Path x \arrow{r}[description]{\Path f} & \Path y
\end{tikzcd}
\end{center}

Furthermore, $\Path \colon \A\to\T$ is a \emph{(Street) fibration} if, for all morphisms in $\T$ of the form $h\colon T\to \Path y$, there exists a Cartesian morphism $f\colon x\to y$ in $\A$ and an isomorphism $\iota\colon \Path x\cong T$ in $\T$ such that $\Path f = h\circ \iota$.
\end{dfn}

The first step in showing that \(\Path\) is a fibration is a technical result that characterises Cartesian morphisms as the so-called ``pathwise embeddings'', which is an important class of morphisms in the theory of arboreal categories.

\begin{theorem}
\label{t:Cartesian-iff-pw}
    A morphism \(f\colon X\to Y\) in \(\A\) is Cartesian if, and only if, it is a \emph{pathwise embedding}, i.e.\ \(f\circ m\) is an embedding for every path embedding \(m\colon P\emb X\).
\end{theorem}

We prove Theorem~\ref{t:Cartesian-iff-pw} in several steps as follows.

\begin{proposition}\label{p:pw-implies-Cartesian}
Every pathwise embedding is Cartesian.
\end{proposition}

\begin{prf}
Let $f\colon X\to Y$ be a pathwise embedding and suppose that we are given a commutative diagram as on the right-hand side below.
\begin{center}
\begin{tikzcd}[column sep = 4em]
Z \arrow[bend left=30]{drr}[description]{g} \arrow[dashed]{dr}[description]{\o{w}} & & \\
 & X \arrow{r}[description]{f} & Y
\end{tikzcd}
\ \ \ 
\begin{tikzcd}[column sep = 3em]
\Path Z \arrow[bend left=30]{drr}[description]{\Path g} \arrow{dr}[description]{w}& & \\
 & \Path X \arrow{r}[description]{\Path f} & \Path Y
\end{tikzcd}
\end{center}
We must show that there is a unique $\o{w}\colon Z\to X$ such that $\Path \o{w} = w$ and $f\circ \o{w} = g$. 

We shall define $\o{w}$ as the unique mediating morphism induced by a compatible cocone 
\[
\{\phi_{n} \colon Q \to X \mid n\colon Q\emb Z \}
\]
over the diagram of path embeddings into $Z$. Fix an arbitrary path embedding ${n\colon Q\emb Z}$ and consider the $(\Q,\M)$ factorisation of $g\circ n$:
\[\begin{tikzcd}
Q \arrow[twoheadrightarrow]{r}{e_{n}} & Q' \arrow[rightarrowtail]{r}{m_{n}} & Y.
\end{tikzcd}\]
Then $[m_{n}] = \Path g (n) = \Path f (w([n]))$ entails the existence of an embedding
\[
    {j_{n}\colon Q'\emb X}
\]
such that $[j_{n}] = w([n])$ and the following diagram commutes (just recall that $f$ is a pathwise embedding). 
\[\begin{tikzcd}[row sep = 2.5em]
Q \arrow[rightarrowtail]{r}{n} \arrow[twoheadrightarrow]{d}[swap]{e_{n}} & Z \arrow{r}{g} & Y \\
Q' \arrow[rightarrowtail]{rr}{j_{n}} \arrow[rightarrowtail]{urr}[description]{m_{n}} & & X \arrow{u}[swap]{f}
\end{tikzcd}\]
Let $\phi_{n} \coloneqq j_{n} \circ e_{n} \colon Q\to X$. 

To see that the morphisms $\phi_{n}$ induce a compatible cocone over the diagram of path embeddings into $Z$, suppose that $\ell \colon R \emb Z$ is another path embedding and there exists $i\colon Q \emb R$ such that $\ell\circ i = n$. We must prove that $\phi_{n} = \phi_{\ell} \circ i$. Consider the following commutative square, and note that the bottom horizontal arrow is an embedding because $f$ is a pathwise embedding, hence there exists a diagonal filler $d$.
\[\begin{tikzcd}[column sep = 2.5em, row sep = 2.5em]
Q \arrow[twoheadrightarrow]{r}{e_{n}} \arrow[swap]{d}{e_{\ell}\circ i} & Q' \arrow[rightarrowtail]{d}{f\circ j_{n}} \arrow[dashed, rightarrowtail]{dl}[description]{d} \\
R' \arrow[rightarrowtail]{r}{f\circ j_{\ell}} & Y
\end{tikzcd}\]
Suppose for a moment that $j_{\ell}\circ d = j_{n}$. Then we have
\[
\phi_{n} = j_{n} \circ e_{n} = j_{\ell}\circ d \circ e_{n} = j_{\ell}\circ e_{\ell} \circ i = \phi_{\ell} \circ i,
\]
showing that the cocone is compatible. To verify the identity $j_{\ell}\circ d = j_{n}$, note that $\ell\circ i = n$ implies $[n]\leq [\ell]$ in $\Path Z$ and so, because $w$ is monotone,
\[
[j_{n}] = w[n] \leq w[\ell] = [j_{\ell}].
\]
This implies that $j_{\ell}\circ d' = j_{n}$ for some embedding $d'$, and thus
\[
f\circ j_{\ell} \circ d' = f \circ j_{n} = f\circ j_{\ell} \circ d.
\]
Since $f\circ j_{\ell}$ is an embedding (hence, a monomorphism), we get $d'=d$. It follows that $j_{\ell}\circ d = j_{n}$.

Write $\o{w}\colon Z \to X$ for the unique mediating morphism induced by the compatible cocone above. We claim that $\Path \o{w} = w$ and $f\circ \o{w} = g$. For the latter identity, observe that, for all path embeddings $n\colon Q\emb Z$,
\[
f\circ \o{w} \circ n = f\circ \phi_{n} = f\circ j_{n} \circ e_{n} = g\circ n
\]
and therefore $f\circ \o{w} = g$ because $Z$ is path-generated. To see that $\Path \o{w} = w$, note that for any $[n]\in \Path Z$
\[
\Path \o{w}([n]) = \exists_{\phi_{n}} \id_{Q} = [j_{n}] = w([n])
\]
where, as is common in the arboreal category literature, \(\exists_{\phi_{n}} \id_{Q}\) denotes the embedding part of the $(\Q,\M)$ factorisation of \(\phi_n\circ \id_Q\).

Finally, we show that $\o{w}$ is unique with these properties. Suppose ${\w{w}\colon Z\to X}$ is such that $\Path\w{w} = w$ and $f\circ \w{w} = g$. Since $Z$ is path-generated, to conclude that $\o{w} = \w{w}$ it suffices to show that $\o{w} \circ n = \w{w}\circ n$ for all path embeddings $n\colon Q\emb Z$. The identity
\[
\Path\w{w}([n]) = w([n]) = [j_{n}]
\]
entails the existence of a quotient $\epsilon_{n}\colon Q\epi Q'$ such that $\w{w}\circ n = j_{n}\circ \epsilon_{n}$. If we can prove that $\epsilon_{n} = e_{n}$ then it will follow that
\[
\o{w} \circ n = \phi_{n} = j_{n} \circ e_{n} = j_{n} \circ \epsilon_{n} = \w{w}\circ n,
\]
as desired. Now, observe that
\[
f\circ j_{n}\circ \epsilon_{n} = f\circ \w{w}\circ n = g\circ n = f\circ \o{w}\circ n = f\circ j_{n}\circ e_{n}.
\]
The composite $f\circ j_{n}$ is an embedding because $f$ is a pathwise embedding, and therefore we conclude that $\epsilon_{n} = e_{n}$.
\end{prf}

To establish the converse of Proposition~\ref{p:pw-implies-Cartesian}, namely that every Cartesian morphism is a pathwise embedding, we start by looking at a special case:
\begin{lemma}\label{l:cartesian-path-emb}
Let $h\colon P\to X$ be a morphism in $\A$ with $P$ a path. If $h$ is Cartesian, then it is an embedding.
\end{lemma}

\begin{prf}
Consider the $(\Q,\M)$ factorisation of $h$:
\[\begin{tikzcd}
P \arrow[twoheadrightarrow]{r}{e} & Q \arrow[rightarrowtail]{r}{m} & X.
\end{tikzcd}\]
The forest morphism $\Path{e}$ is a bijection (because it is a surjective forest morphism between chains), hence an isomorphism in $\T$; we shall write $w$ for its inverse.
Since $h$ is Cartesian, there is a unique arrow $\o{w}$ in $\A$ such that $\Path \o{w} = w$ and the leftmost diagram below commutes.
\begin{center}
\begin{tikzcd}[column sep = 4em]
Q \arrow[bend left=30, rightarrowtail]{drr}[description]{m} \arrow[dashed]{dr}[description]{\o{w}} & & \\
 & P \arrow{r}[description]{h} & X
\end{tikzcd}
\ \ \ 
\begin{tikzcd}[column sep = 3em]
\Path Q \arrow[bend left=30]{drr}[description]{\Path m} \arrow{dr}[description]{w}& & \\
 & \Path P \arrow{r}[description]{\Path h} & \Path X
\end{tikzcd}
\end{center}

To conclude that $e$ is an isomorphism (with inverse $\o{w}$), and so $h$ is an embedding, we use again the fact that $h$ is Cartesian. Let $\w{w}\colon P\to P$ be the unique arrow in $\A$ such that $\Path\w{w}$ is the identity of $\Path P$ and $h\circ \w{w} =h$.
\begin{center}
\begin{tikzcd}[column sep = 4em]
P \arrow[bend left=30]{drr}[description]{h} \arrow[dashed]{dr}[description]{\w{w}} & & \\
 & P \arrow{r}[description]{h} & X
\end{tikzcd}
\ \ \ 
\begin{tikzcd}[column sep = 3em]
\Path P \arrow[bend left=30]{drr}[description]{\Path h} \arrow{dr}[description]{\id}& & \\
 & \Path P \arrow{r}[description]{\Path h} & \Path X
\end{tikzcd}
\end{center}
The identity of $P$ clearly satisfies these conditions, but so does $\o{w}\circ e$. Just observe that
\[
\Path (\o{w}\circ e) = w\circ \Path e = \id_{\Path P}
\]
and
\[
h\circ \o{w}\circ e = m\circ e = h.
\]
Hence $\o{w}\circ e = \id_{P}$. It follows that $e$ is an embedding (in fact, a section) and therefore an isomorphism.
\end{prf}

Combining the previous observations, we obtain a proof of Theorem~\ref{t:Cartesian-iff-pw}.

\begin{prf}[of Theorem~\ref{t:Cartesian-iff-pw}]
One direction is the content of Proposition~\ref{p:pw-implies-Cartesian}. For the other direction, suppose $f\colon X\to Y$ is Cartesian and let $m\colon P\emb X$ be a path embedding. Then~$m$ is Cartesian by Proposition~\ref{p:pw-implies-Cartesian}, and since Cartesian morphisms are closed under composition, $f\circ m$ is also Cartesian. It follows from Lemma~\ref{l:cartesian-path-emb} that $f\circ m$ is a path embedding, and so $f$ is a pathwise embedding. 
\end{prf}

%
We are now ready to state the main result of this section. 

\begin{theorem}
    \label{t:P-fibration}
    The path functor $\Path \colon \A\to\T$ is a Street fibration.
\end{theorem}

\begin{prf}
Consider a forest morphism $h\colon T\to \Path{Y}$ with $T$ a tree. In view of Theorem~\ref{t:Cartesian-iff-pw}, we must exhibit a pathwise embedding $f\colon X\to Y$ in $\A$ and an isomorphism $\iota\colon \Path X\cong T$ in $\T$ such that $\Path f = h\circ \iota$. We shall define $X$ as the colimit of a tree-diagram of paths indexed by $T$, as follows.

For each $i\in T$, denote by $P_{i}$ the domain of any path embedding $n_{i}$ into $Y$ such that $h(i) = [n_{i}]$; note that any two elements of $[n_{i}]$ have isomorphic domains, so $P_{i}$ is unique up to isomorphism. If $i,j\in T$ satisfy $i\leq j$, then $h(i)\leq h(j)$ and so there exists a unique embedding $e\colon P_{i}\emb P_{j}$ such that $n_{i} = n_{j} \circ e$. This way, the assignment $i\mapsto P_{i}$ yields a tree-diagram of paths $F\colon T\to \A$. Let us set $X\coloneqq \colim{F}$, and denote by $m_{i}\colon P_{i}\to X$ the colimit arrows. 

There is an obvious compatible cocone with vertex $Y$ over the diagram $F$, whose component at $i$ is $n_{i}$. Hence, there exists a unique mediating morphism $f\colon X\to Y$ making the following diagram commute.
\[\begin{tikzcd}[row sep =2.4em]
 & Y & \\
 & X \arrow{u}[description]{f} & \\
 P_{i} \arrow[rightarrowtail]{rr} \arrow[rightarrowtail, bend left=30]{uur}[description]{n_{i}} \arrow{ur}[description]{m_{i}} & & P_{j} \arrow[rightarrowtail, bend right=30]{uul}[description]{n_{j}} \arrow{ul}[description]{m_{j}}
\end{tikzcd}\]
Note in particular that, because each $n_{i}$ is an embedding, so is each $m_{i}$. To see that~$f$ is a pathwise embedding, consider any path embedding $m\colon P\emb X$. Since $P$ is tree-connected, there exists a (least) $k\in T$ such that $m$ factors through~$m_{k}$, say $m = m_{k}\circ \lambda$ with $\lambda\colon P\emb P_{k}$. It follows that
\[
f\circ m = f\circ m_{k}\circ \lambda = n_{k}\circ \lambda,
\]
which is an embedding, and so $f$ is a pathwise embedding.

Write $\rho\colon T \to \Path{X}$ for the map sending $i$ to (the equivalence class of) $m_{i}$, and note that for all $i\in T$ we have $(\Path{f}\circ \rho)(i) = \Path{f}([m_{i}]) = [f\circ m_{i}] = [n_{i}]= h(i)$, that is
\begin{equation}\label{eq:rho-factors}
\Path{f}\circ \rho = h.
\end{equation}
To conclude, it suffices to show that~$\rho$ is an isomorphism of forests, for then its inverse $\iota\coloneqq\rho^{-1}$ satisfies the required property. Observe that eq.~\eqref{eq:rho-factors} implies that~$\rho$ preserves the height of elements. But $\rho$ is also monotone (by the same argument showing functoriality of $F$), therefore it is a forest morphism. To show that $\rho$ is injective, suppose that $\rho(i)\leq \rho(j)$ and let $k\in T$ be the least element such that $m_{i}$ factors through the colimit arrow $m_{k}$; in particular, $k\leq i,j$. 
Recalling from Section~\ref{s:coalg-concretely} that for an element $x$ of a forest, its height is denoted by~\(\htf(x)\),
\[
\htf([m_{k}]) = \htf(\rho(k)) = \htf(k)\leq \htf(i) = \htf(\rho(i)) = \htf([m_{i}])
\]
because $k\leq i$. It follows that $\htf(k) = \htf(i)$, and so $i=k$. Thus, $i\leq j$.

Finally, since the image of $\rho$ is downwards closed, to prove that $\rho$ is surjective it suffices to show that each element of $\Path{X}$ is below one of the form $\rho(k)$, for some $k\in T$.
But this follows by reasoning as above, using the fact that paths are tree-connected.
\end{prf}

One might ask whether Theorem~\ref{t:P-fibration} could be strengthened to assert that \(\Path\) is a topological functor in the sense of \cite[Chapter~21]{adamek2004abstract}. The following example shows that this is not possible.

\begin{example}
    \label{e:not-topological}
    Let \(\A \coloneqq \EM(\Pk)\) for \(k\geq 2\).
    Recall from Section~\ref{s:coalg-concretely} that objects in this category can be described as tuples \((A,\leq, p)\) where \((A,\leq)\) is a forest-ordered structure and \(p\colon A \to \{1,2\}\) is a pebbling function. Consider \(A_1, A_2\) discrete structures on \(\{\star\}\) equipped with \(p\colon A_i \to \{1,2\}\) such that \(p(\star) = i\).
    Then, both \(\Path A_1\) and \(\Path A_2\) are isomorphic to the two-element chain \(\mathbf 2 = (\bot < \top)\), whose bottom element corresponds to the empty subpath of \(A_1\) and \(A_2\), respectively.

    However, the span of forest morphisms \((f_i\colon \mathbf 2 \to \Path A_i)_{i=1,2}\) in \(\T\) cannot be lifted to a span \((g_i\colon X \to A_i)_{i=1,2}\) in \(\A\) such that \(\Path(g_i) = f_i\), since \(X\) would have to be a structure \((\{\star\}, \leq, p)\) such that \(p\colon \{\star\} \to \{1,2\}\) maps \(\star\) to both $1$ and $2$ at the same time.
    In particular, there is no initial lift of the span consisting of the \(f_i\)'s, and so \(\Path\colon \EM(\Pk)\to \T\) is not topological.
\end{example}

It remains an open problem to identify which Street fibrations over the category $\T$ of trees arise from arboreal categories.
Finally, we present an example of a category that is arboreal in the old sense whose path functor is not a fibration.

\begin{example}
Let \(\sg\) be a non-empty relational signature consisting only of unary relations, which we regard as colours. Let \(\A\) be the category \(\RS\), equipped with the factorisation system (regular epi, mono), rather than the factorisation system (epi, regular mono) considered earlier. Then $\A$ is arboreal in the old sense and the path functor \({\Path\colon \A \to \T}\) is not a fibration.

For the sake of this example, assume $\sg$ contains two colours, \emph{blue} and \emph{red}.
The non-empty paths in $\A$ are the one-element structures with at most one colour.
Hence, up to isomorphism, there is one path of height 0 (the empty structure), one path of height 1 (the one-element structure with no colours), and two paths \(P_{blue}\) and \(P_{red}\) of height 2, one for each colour. There are no paths of height 3 or more.
The category \(\A\) is arboreal in the old sense, but not in the new sense, since the tree-diagram of paths \(P_{blue} \leftarrow \{\star\} \to P_{blue}\) has colimit \(P_{blue}\), but the identity \(P_{blue} \to P_{blue}\) does not factor through a minimal element.

Furthermore, the path functor \(\Path\colon \A \to \T\) is not a fibration because the trees in the image of \(\Path\) are coproducts of trees of the shapes below.

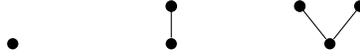
\begin{center}
    \begin{tikzpicture}[
            level distance=5mm,
            every node/.style={fill=black,circle,inner sep=1.5pt},
            level 1/.style={sibling distance=8mm}]
        \begin{scope}
        \node {} [grow'=up]
            child {[fill] circle (2pt)}
            child {[fill] circle (2pt)};
        \end{scope}

        \begin{scope}[xshift=-6em]
        \node {} [grow'=up]
            child {[fill] circle (2pt)};
        \end{scope}

        \begin{scope}[xshift=-12em]
        \node {} [grow'=up];
        \end{scope}
    \end{tikzpicture}
\end{center}
If \(T\) is the tree consisting of a root with three children, there is no \(Y\) such that \(T \cong \Path(Y)\). Thus, no morphism of the form \(T \to \Path X\) admits a Cartesian lift.
\end{example}

\subsubsection{Future work.} We mention some current lines of research that contribute to the structure theory of arboreal categories, in the newly defined sense.
\begin{itemize}
\item The same authors are working towards a representation theorem for \emph{concrete arboreal categories},\footnote{An arboreal category is \emph{concrete (over trees)} if the path functor is faithful.} whereby every such category is equivalent to a category of ``labelled trees''.
\item The relationship between bisimilarity and behaviour equivalence in arboreal categories is being investigated in ongoing research by the first author on model-theoretic types. Objects \(X\) and \(Y\) are said to be \emph{behaviourally equivalent} if there is a cospan of open morphisms \(X \to Z \leftarrow Y\).
\item The existence of arboreal coreflections for certain categories of coalgebras for endofunctors is investigated in an ongoing work by the first author, Henning Urbat and Thorsten Wißmann.
\end{itemize}

\begin{credits}
\subsubsection{\ackname}
The first author has received funding from the EU's Horizon Europe research and innovation programme under the Marie Sk\l{}odowska-Curie grant agreement No~101111373.
The authors would like to thank Henning Urbat for encouraging them to publish this work and for asking the question whether \(\Path\) is a topological functor, answered in Example~\ref{e:not-topological}.
\end{credits}

\bibliographystyle{splncs04}

\end{document}